\documentclass[12pt,reqno]{amsart}
\usepackage{txfonts,fullpage}

\newtheorem{thm}{Theorem}[section]

\newtheorem{lem}[thm]{Lemma}

\newtheorem{rem}[thm]{Remark}
\numberwithin{equation}{section}
\allowdisplaybreaks

\title{A reformulation of the generalized $q$-Painlev\'{e} VI system with $W(A^{(1)}_{2n+1})$ symmetry}
\date{}
\author{Takao Suzuki}
\address{Department of Mathematics, Kindai University, 3-4-1, Kowakae, Higashi-Osaka, Osaka 577-8502, Japan}
\email{suzuki@math.kindai.ac.jp}

\begin{document}

\maketitle

\begin{abstract}
In the previous work we introduced the higher order $q$-Painlev\'{e} system $q$-$P_{(n+1,n+1)}$ as a generalization of the Jimbo-Sakai's $q$-Painlev\'{e} VI equation.
It is derived from a $q$-analogue of the Drinfeld-Sokolov hierarchy of type $A^{(1)}_{2n+1}$ and admits a particular solution in terms of the Heine's $q$-hypergeometric function ${}_{n+1}\phi_n$.
However the obtained system is insufficient as a generalization of $q$-$P_{\rm{VI}}$ due to some reasons.
In this article we rewrite the system $q$-$P_{(n+1,n+1)}$ to a more suitable one.

Key Words: Discrete Painlev\'{e} equations, Affine Weyl groups, Basic hypergeometric functions.

2010 Mathematics Subject Classification: 39A13, 33D15, 34M55.
\end{abstract}

\section{Introduction}\label{sec:Intro}

Several generalizations of the Painlev\'e VI equation ($P_{\rm{VI}}$) have been proposed (\cite{FS1,G,K,Sak2,Sas,Su2,Su3,Tsu2}).
We focus on the higher order Painlev\'e system $P_{(n+1,n+1)}$ given in \cite{FS1,Su2}, or equivalently the Schlesinger system $\mathcal{H}_{n+1,1}$ given in \cite{Tsu2}, among them.
It can be regarded as a generalization from a viewpoint of a particular solution in terms of the hypergeometric function ${}_{n+1}F_n$ (\cite{Su1,Tsu3}).
The aim of this article is to introduce its $q$-analogue.
This $q$-difference equation becomes a generalization of the $q$-Painlev\'{e} VI equation ($q$-$P_{\rm{VI}}$) given in \cite{JS}.

The investigation of generalizations of $q$-$P_{\rm{VI}}$ has been developed in recent years (\cite{M2,NY,Sak1,Su4,Tsu1}).
In the previous work \cite{Su4} we proposed the higher order $q$-Painlev\'e system $q$-$P_{(n+1,n+1)}$, whose explicit formula will be given in Section \ref{sec:Review_qFST}, as a $q$-analogue of $P_{(n+1,n+1)}$.
It is derived from the $q$-Drinfeld-Sokolov hierarchy of type $A^{(1)}_{2n+1}$, contains $q$-$P_{\rm{VI}}$ in the case of $n=1$ and admits a particular solution in terms of the $q$-hypergeometric function ${}_{n+1}\phi_n$.
However this system is insufficient as a generalization of $q$-$P_{\rm{VI}}$ due to the following two reasons.

\begin{enumerate}
\item
The system $q$-$P_{(n+1,n+1)}$ is probably reducible and reduces to a system of $2n$-th order.
\item
We do not express the backward $q$-shifts $(\underline{x_i},\underline{y_i})$ as functions in $(x_i,y_i)$.
\end{enumerate}

The aim of this article is to solve those two problems.
We reduce the system $q$-$P_{(n+1,n+1)}$ to a more suitable one as a generalization of $q$-$P_{\rm{VI}}$.

This article is organized as follows.
In Section \ref{sec:Review_qFST} we recall the definition of $q$-$P_{(n+1,n+1)}$ and its properties, namely a Lax pair, an affine Weyl group symmetry, a relationship with $q$-$P_{\rm{VI}}$ and a particular solution in terms of ${}_{n+1}\phi_n$.
In Section \ref{sec:Main_result} we formulate a system of $q$-difference equations of $2n$-th order which is equivalent to $q$-$P_{\rm{VI}}$ in the case of $n=1$.
It is the main result of this article.
In Section \ref{sec:Affine_Weyl} we describe an action of the affine Weyl group on the $2n$-th order system given in the previous section.

\section{Review: Higher order $q$-Painlev\'{e} system $q$-$P_{(n+1,n+1)}$}\label{sec:Review_qFST}

In the following we use notations
\[
	\overline{x(t)} = x(qt),\quad
	\underline{x(t)} = x(q^{-1}t),
\]
where $t,q\in\mathbb{C}$ and $|q|<1$.

The system $q$-$P_{(n+1,n+1)}$ given in \cite{Su4} is expressed as a system of $q$-difference equations
\[
	\left\{\begin{array}{l}
		\displaystyle x_{i-1} - x_i = \frac{b_{i-1}\underline{x_{i-1}}}{1+\underline{x_{i-1}}y_{i-1}} - \frac{a_i\underline{x_i}}{1+\underline{x_i}y_{i-1}} \\[12pt]
		\displaystyle \underline{y_{i-1}} - \underline{y_i} = \frac{a_iy_{i-1}}{1+\underline{x_i}y_{i-1}} - \frac{b_iy_i}{1+\underline{x_i}y_i}
	\end{array}\right.\quad (i=1,\ldots,n+1),
\]
with a constraint
\begin{equation}\label{eq:qFST_constraint}
	\prod_{i=1}^{n+1}a_i\frac{1+\underline{x_i}y_i}{1+\underline{x_i}y_{i-1}} = q^{-n/2},
\end{equation}
where
\[
	b_0 = qb_{n+1},\quad x_0 = tx_{n+1},\quad y_0 = \frac{q}{t}y_{n+1}.
\]

\begin{rem}
In the previous work \cite{Tsu1} a higher order generalizations of $q$-$P_{\rm{VI}}$ were presented by Tsuda.
Since his $q$-Painlev\'e system can be regarded as a $q$-analogue of the Schlesinger system $\mathcal{H}_{n+1,1}$, we conjecture that his system coinsides with $q$-$P_{(n+1,n+1)}$.
However a relationship between both $q$-Painlev\'e systems has not been clarified yet.
\end{rem}

We derived the system $q$-$P_{(n+1,n+1)}$ by a similarity reduction from the $q$-Drinfeld-Sokolov hierarchy of type $A^{(1)}_{2n+1}$.
Hence the following theorem is obtained naturally via the construction of the system.

\begin{thm}[\cite{Su4}]
The system $q$-$P_{(n+1,n+1)}$ is given as the compatibility conditon of a system of linear $q$-difference equations 
\begin{equation}\label{eq:Lax_qFST}
	\Psi(q^{-1}z,t) = M(z,t)\Psi(z,t),\quad
	\Psi(z,q^{-1}t) = B(z,t)\Psi(z,t).
\end{equation}
with $(2n+2)\times(2n+2)$ matrices
\[
	M(z,t) = \begin{bmatrix}a_1&\varphi_1&-1\\&b_1&\varphi_2&-1\\&&a_2&\varphi_3&-1\\&&&b_2&\varphi_4\\&&&&&\ddots\\&&&&&&\varphi_{2n-1}&-1\\&&&&&&b_n&\varphi_{2n}&-1\\-tz&&&&&&&a_{n+1}&\varphi_{2n+1}\\\varphi_0z&-z&&&&&&&b_{n+1}\end{bmatrix},
\]
and
\[
	B(z,t) = \begin{bmatrix}u_1&v_1&-1\\&u_2&v_2&0\\&&u_3&v_3&-1\\&&&u_4&v_4\\&&&&&\ddots\\&&&&&&v_{2n-1}&-1\\&&&&&&u_{2n}&v_{2n}&0\\-tz&&&&&&&u_{2n+1}&v_{2n+1}\\v_0z&0&&&&&&&u_{2n+2}\end{bmatrix},
\]
where
\[\begin{split}
	&\varphi_{2i-2} = x_{i-1} - x_i,\quad
	\varphi_{2i-1} = y_{i-1} - y_i\quad (i=1,\ldots,n+1), \\
	&u_{2i-1} = \frac{a_i}{1+\underline{x_i}y_{i-1}},\quad
	u_{2i} = 1 + \underline{x_i}y_i,\quad
	v_{2i-1} = -y_i\quad (i=1,\ldots,n+1), \\
	&v_0 = t\underline{x_{n+1}},\quad
	v_{2i} = \underline{x_i}\quad (i=1,\ldots,n).
\end{split}\]
\end{thm}

The system $q$-$P_{(n+1,n+1)}$ admits the affine Weyl group symmetry of type $A^{(1)}_{2n+1}$.

\begin{thm}
Let $r_0,\ldots,r_{2n+1}$ be birational transformations defined by
\begin{equation}\begin{split}\label{eq:Aff_Wey_even}
	&r_{2j-2}(a_j) = b_{j-1},\quad
	r_{2j-2}(b_{j-1}) = a_j,\quad
	r_{2j-2}(x_{j-1}) = x_{j-1},\quad
	r_{2j-2}(y_{j-1}) = y_{j-1} - \frac{b_{j-1}-a_j}{x_{j-1}-x_j}, \\
	&r_{2j-2}(a_i) = a_i,\quad
	r_{2j-2}(b_{i-1}) = b_{i-1},\quad
	r_{2j-2}(x_{i-1}) = x_{i-1},\quad
	r_{2j-2}(y_{i-1}) = y_{i-1}\quad (i\neq j),
\end{split}\end{equation}
for $j=1,\ldots,n+1$, and
\begin{equation}\begin{split}\label{eq:Aff_Wey_odd}
	&r_{2j-1}(a_j) = b_j\quad
	r_{2j-1}(b_j) = a_j,\quad
	r_{2j-1}(x_j) = x_j - \frac{a_j-b_j}{y_{j-1}-y_j},\quad
	r_{2j-1}(y_j) = y_j, \\
	&r_{2j-1}(a_i) = a_i,\quad
	r_{2j-1}(b_i) = b_i,\quad
	r_{2j-1}(x_i) = x_i,\quad
	r_{2j-1}(y_i) = y_i\quad (i\neq j),
\end{split}\end{equation}
for $j=1,\ldots,n+1$.
Also let $\pi$ be a birational transformation defined by
\begin{equation}\begin{split}\label{eq:Aff_Wey_rot}
	&\pi(a_i) = q^{-\rho_1}b_i,\quad
	\pi(b_i) = q^{-\rho_1}a_{i+1}\quad (i=1,\ldots,n), \\
	&\pi(a_{n+1}) = q^{-\rho_1}b_{n+1},\quad
	\pi(b_{n+1}) = q^{-\rho_1-1}a_1,\quad
	\pi(\rho_1) = -\rho_1-\frac{1}{n+1}, \\
	&\pi(x_i) = q^{-2\rho_1}t^{\rho_1}\underline{y_i},\quad
	\pi(y_i) = q^{\rho_1}t^{-\rho_1}\underline{x_{i+1}}\quad (i=1,\ldots,n), \\
	&\pi(x_{n+1}) = q^{-2\rho_1}t^{\rho_1}\underline{y_{n+1}},\quad
	\pi(y_{n+1}) = q^{\rho_1+1}t^{-\rho_1-1}\underline{x_1},\quad
	\pi(t) = \frac{q^2}{t},
\end{split}\end{equation}
where
\[
	q^{\rho_1} = (q^na_1b_1\ldots a_{n+1}b_{n+1})^{1/(n+1)}.
\]
Then the system $q$-$P_{(n+1,n+1)}$ is invariant under actions of the transformations $r_0,\ldots,r_{2n+1}$ and $\pi$.
Furthermore the group of symmetries $\langle r_0,\ldots,r_{2n+1},\pi\rangle$ is isomorphic to the extended affine Weyl group of type $A_{2n+1}^{(1)}$.
Namely those transformations satisfy the fundamental relations
\[\begin{split}
	&r_i^2=1,\quad
	(r_ir_j)^{2-a_{i,j}}=1\quad (i,j=0,\ldots,2n+1; i\neq j), \\
	&\pi^{2n+2} = 1,\quad
	\pi r_i = r_{i+1}\pi,\quad
	\pi r_{2n+1} = r_0\pi\quad (i=0,\ldots,2n),
\end{split}\]
where
\[\begin{array}{llll}
	a_{i,i}=2 & (i=0,\ldots,2n+1), \\[4pt]
	a_{i,i+1}=a_{2n+1,0}=a_{i+1,i}=a_{0,2n+1}=-1 & (i=0,\ldots,2n), \\[4pt]
	a_{i,j}=0 & (\text{otherwise}).
\end{array}\]
\end{thm}

The system $q$-$P_{(2,2)}$ can be reduced to $q$-$P_{\rm{VI}}$.

\begin{thm}[\cite{Su4}]
If, in the system $q$-$P_{(2,2)}$, we set
\begin{equation}\label{eq:def_fg_qP6}
	f = \frac{t(x_2-x_1)\xi_1}{\xi_2},\quad
	g = \frac{\underline{x_2}(qt+\underline{x_1}y_2)\psi_1}{(1+\underline{x_2}y_2)\psi_2},
\end{equation}
where
\[\begin{split}
	\xi_1 &= (x_1-x_2)(y_0-y_1) - (a_1-b_1), \\
	\xi_2 &= (tx_2-x_1)(x_1-x_2)(y_0-y_1) + (a_1-b_1)x_1 + \{(b_1-a_2)t-(a_1-a_2)\}x_2, \\
	\psi_1 &= q^{1/2}(q^{1/2}-a_1b_1t)\underline{x_2}y_2 + (1-q^{1/2}a_1b_1)t, \\
	\psi_2 &= q^{1/2}a_2(q^{1/2}-a_1b_1t)\underline{x_1}\underline{x_2}y_2 + a_1(1-q^{1/2}b_1a_2)t\underline{x_1} - (a_1-a_2)t\underline{x_2},
\end{split}\]
then they satisfy the $q$-Painlev\'{e} VI equation
\[
	\frac{f\overline{f}}{\alpha_3\alpha_4} = \frac{(\overline{g}-t\beta_1)(\overline{g}-t\beta_2)}{(\overline{g}-\beta_3)(\overline{g}-\beta_4)},\quad
	\frac{g\overline{g}}{\beta_3\beta_4} = \frac{(f-t\alpha_1)(f-t\alpha_2)}{(f-\alpha_3)(f-\alpha_4)},
\]
with parameters
\[\begin{split}
	&\alpha_1 = 1,\quad
	\alpha_2 = q^{1/2}a_1b_1,\quad
	\alpha_3 = 1,\quad
	\alpha_4 = \frac{1}{q^{1/2}a_2b_2}, \\
	&\beta_1 = q^{1/2}b_1,\quad
	\beta_2 = q^{1/2}a_1,\quad
	\beta_3 = \frac{1}{qa_2},\quad
	\beta_4 = \frac{1}{b_2}.
\end{split}\]
\end{thm}

\begin{rem}
In \cite{Su4} we defined the transformation $\pi$ by
\[\begin{split}
	&\pi(a_i) = b_i,\quad
	\pi(b_i) = a_{i+1}\quad (i=1,\ldots,n),\quad
	\pi(a_{n+1}) = b_{n+1},\quad
	\pi(b_{n+1}) = \frac{a_1}{q}, \\
	&\pi(x_i) = \underline{y_i},\quad
	\pi(y_i) = \underline{x_{i+1}}\quad (i=1,\ldots,n),\quad
	\pi(x_{n+1}) = \underline{y_{n+1}},\quad
	\pi(y_{n+1}) = \frac{q}{t}\underline{x_1},\quad
	\pi(t) = \frac{q^2}{t}.
\end{split}\]
As a matter of fact, unless we replace constraint \eqref{eq:qFST_constraint} with
\[
	\prod_{i=1}^{n+1}\frac{a_i^{1/2}}{b_i^{1/2}}\frac{1+\underline{x_i}y_i}{1+\underline{x_i}y_{i-1}} = q^{1/4},
\]
the system $q$-$P_{(n+1,n+1)}$ is not invariant under an action of $\pi$.
If we do so, then the system $q$-$P_{(2,2)}$ seems to reduce not to $q$-$P_{\rm{VI}}$ but to a $q$-analogue of the Painlev\'e V equation.
This $q$-difference equation was derived in \cite{M1} from a binational representation of the extended affine Weyl group of type $A^{(1)}_1\times A^{(1)}_3$ given in \cite{KNY}.
Afterward that equation was found to be a subsystem of $q$-$P_{\rm{VI}}$ in \cite{Tak}.
\end{rem}

The system $q$-$P_{(n+1,n+1)}$ admits a particular solution in terms of the $q$-hypergeometric function ${}_{n+1}\phi_n$ defined by the formal power series
\[
	{}_{n+1}\phi_n\left[\begin{array}{c}\alpha_1,\ldots,\alpha_n,\alpha_{n+1}\\\beta_1,\ldots,\beta_n\end{array};q,t\right]
	= \sum_{k=0}^{\infty}\frac{(\alpha_1;q)_k\ldots(\alpha_n;q)_k(\alpha_{n+1};q)_k}{(\beta_1;q)_k\ldots(\beta_n;q)_k(q;q)_k}t^k,
\]
where $(\alpha;q)_k$ stands for the $q$-shifted factorial
\[
	(\alpha;q)_0 = 1,\quad
	(\alpha;q)_k = (1-\alpha)(1-q\alpha)\ldots(1-q^{k-1}\alpha)\quad (k\geq1).
\]

\begin{thm}[\cite{Su4}]
If, in the system $q$-$P_{(n+1,n+1)}$, we assume that
\[
	y_i = 0\quad (i=1,\ldots,n+1),\quad \prod_{i=1}^{n+1}a_i = q^{-n/2},
\]
then a vector of the variables $\mathbf{x}={}^t[x_1,\ldots,x_{n+1}]$ satisfies a system of linear $q$-difference equations
\begin{equation}\label{eq:qHGE}
	\overline{\mathbf{x}} = \left(A_0+\frac{A_1}{1-qt}\right)\mathbf{x},
\end{equation}
with $(n+1)\times(n+1)$ matrices
\[\begin{split}
	A_0 &= \begin{bmatrix}b_1&b_2-a_2&b_3-a_3&\ldots&b_n-a_n&b_{n+1}-a_{n+1}\\&b_2&b_3-a_3&\ldots&b_n-a_n&b_{n+1}-a_{n+1}\\&&&\ddots&\vdots&\vdots\\&&&&b_n-a_n&b_{n+1}-a_{n+1}\\&O&&&b_n&b_{n+1}-a_{n+1}\\&&&&&b_{n+1}\end{bmatrix}, \\
	A_1 &= \begin{bmatrix}1\\1\\\vdots\\1\\1\\1\end{bmatrix}\begin{bmatrix}a_1-b_1&a_2-b_2&a_3-b_3&\ldots&a_n-b_n&a_{n+1}-b_{n+1}\end{bmatrix}.
\end{split}\]
Furthermore system \eqref{eq:qHGE} admits a solution
\[\begin{split}
	\mathbf{x} &= t^{-\log_qa_1}\begin{bmatrix}\phi_1\\\vdots\\\phi_{n+1}\end{bmatrix}, \\
	\phi_j &= \prod_{i=1}^{j-1}\frac{b_i-a_1}{a_{i+1}-a_1}{}_{n+1}\phi_n\left[\begin{array}{c}q\frac{a_1}{b_1},\ldots,q\frac{a_1}{b_{j-1}},\frac{a_1}{b_j},\ldots,\frac{a_1}{b_n},\frac{a_1}{b_{n+1}}\\q\frac{a_1}{a_2},\ldots,q\frac{a_1}{a_j},\frac{a_1}{a_{j+1}},\ldots,\frac{a_1}{a_n}\end{array};q,q^{(n+2)/2}b_1\ldots b_{n+1}t\right].
\end{split}\]
\end{thm}

Therefore we want to regard the system $q$-$P_{(n+1,n+1)}$ as a generalization of $q$-$P_{\rm{VI}}$ from a viewpoint of a particular solution in terms of the $q$-hypergeometric function.
However, as is seen in Section \ref{sec:Intro}, we have two problems.
In \cite{Su4} we derived $q$-$P_{\rm{VI}}$ by reducing linear system \eqref{eq:Lax_qFST} to the one with $2\times2$ matrices given in \cite{JS} from $q$-$P_{(2,2)}$.
We could not use a similar method in a general case, although we reduced linear system \eqref{eq:Lax_qFST} to the one with $(n+1)\times(n+1)$ matrices in \cite{FS2}.
Definition of dependent variables \eqref{eq:def_fg_qP6} is complicated and hence unsuitable for a generalization.
In this article we choose a suitable set of $2n$ dependent variables $(f_i,g_i)$ and reduce the system $q$-$P_{(n+1,n+1)}$ to a one of $2n$-th order.
In the obtained system the forward $q$-shifts $(\overline{f}_i,\overline{g}_i)$ is expressed as functions in $(f_i,g_i)$.

\section{Main result}\label{sec:Main_result}

The key to solving is the affine Weyl group symmetry of $q$-$P_{(n+1,n+1)}$.
We can simplify definition of dependent variables \eqref{eq:def_fg_qP6} as
\[
	r_1r_2(f)	= -t\frac{x_1-x_2}{tx_2-x_1},\quad
	r_1r_2(g)	= \frac{a_2t}{q^{1/2}}\frac{\underline{x_2}(1+\underline{x_1}y_1)}{\underline{x_1}(1+\underline{x_2}y_1)}.
\]
This fact suggests a choice of dependent variables of a $2n$-th order system.

\begin{thm}\label{thm:qFST}
If, in the system $q$-$P_{(n+1,n+1)}$, we set
\begin{equation}\label{eq:def_fg}
	f_i = t\frac{x_i-x_{i+1}}{tx_{n+1}-x_1},\quad
	g_i = a_{i+1}\frac{\underline{x_{i+1}}(1+\underline{x_i}y_i)}{\underline{x_i}(1+\underline{x_{i+1}}y_i)}\quad (i=1,\ldots,n),
\end{equation}
then they satisfy a system of $q$-difference equations
\begin{align}
	&f_i\overline{f}_i = qt\frac{F_iF_{i+1}\overline{g}_0(b_i-\overline{g}_i)(\overline{g}_i-a_{i+1})}{F_{n+1}F_1\overline{g}_i(b_0-\overline{g}_0)(\overline{g}_0-a_1)} &(i=1,\ldots,n), \label{eq:qFST_f} \\
	&g_i\overline{g}_i = \frac{F_{i+1}G_i}{F_iG_{i+1}} &(i=1,\ldots,n), \label{eq:qFST_g}
\end{align}
where
\[
	b_0 = qb_{n+1},\quad
	g_0 = \frac{1}{q^{(n-2)/2}t}\prod_{i=1}^{n}\frac{1}{g_i} = \frac{qa_1}{t}\frac{\underline{x_1}(1+\underline{x_{n+1}}y_{n+1})}{\underline{x_{n+1}}(1+\underline{x_1}y_0)} = a_1\frac{\underline{x_1}(1+\underline{x_0}y_0)}{\underline{x_0}(1+\underline{x_1}y_0)},
\]
and
\[\begin{split}
	F_i &= \sum_{j=1}^{i-1}f_j + t\sum_{j=i}^{n}f_j + t, \\
	G_i &= \sum_{j=i}^{n}\prod_{k=i}^{j-1}b_ka_{k+1}\frac{\prod_{l=j+1}^{n}g_l}{\prod_{l=1}^{j-1}g_l}f_j + q^{n/2}t\prod_{k=i}^{n}b_ka_{k+1} \\
	&\quad+ q^nt\sum_{j=1}^{i-1}\frac{b_{n+1}a_1\prod_{k=1}^{n}b_ka_{k+1}}{\prod_{k=j}^{i-1}b_ka_{k+1}}\frac{\prod_{l=j+1}^{n}g_l}{\prod_{l=1}^{j-1}g_l}f_j\quad (i=1,\ldots,n+1).
\end{split}\]
Furthermore if we set $h=tx_{n+1}-x_1$, then it satisfies a linear $q$-difference equation
\begin{equation}\label{eq:qFST_h}
	\overline{h} = -\frac{F_{n+1}F_1(b_0-\overline{g}_0)(\overline{g}_0-a_1)}{t(t-1)^2\overline{g}_0}h.
\end{equation}

\end{thm}

System \eqref{eq:qFST_f}, \eqref{eq:qFST_g} is equivalent to $q$-$P_{\rm{VI}}$ in the case of $n=1$.
In this section we prove this theorem.
We also discuss the relationship between the dependent variables $f_i,g_i,h$ $(i=1,\ldots,n)$ and the ones $x_j,y_j$ $(j=1,\ldots,n+1)$ in more detail at the end of this section.

\subsection{Proof of equation \eqref{eq:qFST_f}}

Definition of dependent variables \eqref{eq:def_fg} implies
\begin{equation}\label{eq:qFST_f_a}
	g_i - a_{i+1} = -a_{i+1}\frac{\underline{x_i}-\underline{x_{i+1}}}{\underline{x_i}(1+\underline{x_{i+1}}y_i)}\quad (i=0,1,\ldots,n).
\end{equation}
The first equation of $q$-$P_{(n+1,n+1)}$ can be rewritten to
\begin{equation}\label{eq:qFST_f_b}
	b_i - g_i = \frac{1+\underline{x_i}y_i}{\underline{x_i}}(x_i-x_{i+1})\quad (i=0,1,\ldots,n).
\end{equation}
Combining them, we obtain
\begin{equation}\begin{split}\label{eq:qFST_f_ab}
	\frac{(b_0-g_0)(g_0-a_1)}{g_0} &= -\frac{(\underline{tx_{n+1}}-\underline{x_1})(tx_{n+1}-x_1)}{\underline{tx_{n+1}}\underline{x_1}}, \\
	\frac{(b_i-g_i)(g_i-a_{i+1})}{g_i} &= -\frac{(\underline{x_i}-\underline{x_{i+1}})(x_i-x_{i+1})}{\underline{x_i}\underline{x_{i+1}}}\quad (i=1,\ldots,n).
\end{split}\end{equation}
Hence we can derive equation \eqref{eq:qFST_f} by using
\[
	F_i = \frac{t(t-1)x_i}{tx_{n+1}-x_1}\quad (i=1,\ldots.n+1).
\]

\subsection{Proof of equations \eqref{eq:qFST_g} and \eqref{eq:qFST_h}}

Equation \eqref{eq:qFST_f_a} can be rewritten to
\[\begin{split}
	y_0 &= -\frac{\underline{x_0}g_0-a_1\underline{x_1}}{\underline{x_0}\underline{x_1}(g_0-a_1)} = -\frac{t\underline{x_{n+1}}g_0-qa_1\underline{x_1}}{t\underline{x_{n+1}}\underline{x_1}(g_0-a_1)} = \frac{q}{t}y_{n+1}, \\
	y_i &= -\frac{\underline{x_i}g_i-a_{i+1}\underline{x_{i+1}}}{\underline{x_i}\underline{x_{i+1}}(g_i-a_{i+1})}\quad (i=1,\ldots,n).
\end{split}\]
Substituting it to the second equation of $q$-$P_{(n+1,n+1)}$, we obtain
\[\begin{split}
	\underline{y_{i-1}} - \underline{y_i} &= \frac{\underline{x_{i-1}}g_{i-1}}{\underline{x_i}(\underline{x_{i-1}}-\underline{x_i})} - \frac{a_i}{\underline{x_{i-1}}-\underline{x_i}} + \frac{b_ia_{i+1}\underline{x_{i+1}}}{\underline{x_i}(\underline{x_i}-\underline{x_{i+1}})g_i} - \frac{b_i}{\underline{x_i}-\underline{x_{i+1}}}\quad (i=1,\ldots,n), \\
	\underline{y_n} - \underline{y_{n+1}} &= \frac{\underline{x_n}g_n}{\underline{x_{n+1}}(\underline{x_n}-\underline{x_{n+1}})} - \frac{a_{n+1}}{\underline{x_n}-\underline{x_{n+1}}} + \frac{qb_{n+1}a_1\underline{x_1}}{\underline{x_{n+1}}(t\underline{x_{n+1}}-q\underline{x_1})g_0} - \frac{b_{n+1}t}{t\underline{x_{n+1}}-q\underline{x_1}}.
\end{split}\]
It implies
\begin{equation}\begin{split}\label{eq:qFST_g_lem_proof_1}
	&b_i(x_{i-1}-x_i) + a_i(x_i-x_{i+1}) + (x_{i-1}-x_i)(x_i-x_{i+1})(y_{i-1}-y_i) \\
	&= \frac{x_{i-1}}{x_i}(x_i-x_{i+1})\overline{g_{i-1}} + b_ia_{i+1}\frac{x_{i+1}}{x_i}(x_{i-1}-x_i)\frac{1}{\overline{g_i}}\quad (i=1,\ldots,n), \\
	&b_{n+1}t(x_n-x_{n+1}) + a_{n+1}(tx_{n+1}-x_1) + (x_n-x_{n+1})(tx_{n+1}-x_1)(y_n-y_{n+1}) \\
	&= \frac{x_n}{x_{n+1}}(tx_{n+1}-x_1)\overline{g_n} + b_{n+1}a_1\frac{x_1}{x_{n+1}}(x_n-x_{n+1})\frac{1}{\overline{g_0}}.
\end{split}\end{equation}
On the other hand, equation \eqref{eq:qFST_f_b} can be rewritten to
\[\begin{split}
	y_0 &= -\frac{g_0-b_0}{x_0-x_1} - \frac{1}{\underline{x_0}} = -\frac{g_0-qb_{n+1}}{tx_{n+1}-x_1} - \frac{q}{t\underline{x_{n+1}}} = \frac{q}{t}y_{n+1}, \\
	y_i &= -\frac{g_i-b_i}{x_i-x_{i+1}} - \frac{1}{\underline{x_i}}\quad (i=1,\ldots,n).
\end{split}\]
Combining it with equation \eqref{eq:qFST_f_ab}, we obtain
\[\begin{split}
	y_{i-1} - y_i &= \frac{g_i-b_i}{x_i-x_{i+1}} + \frac{1}{\underline{x_i}} - \frac{g_{i-1}-b_{i-1}}{x_{i-1}-x_i} - \frac{1}{\underline{x_{i-1}}} \\
	&= \frac{g_i}{x_i-x_{i+1}} - \frac{b_i}{x_i-x_{i+1}} + \frac{b_{i-1}a_i}{(x_{i-1}-x_i)g_{i-1}} - \frac{a_i}{x_{i-1}-x_i}\quad (i=1,\ldots,n), \\
	y_n - y_{n+1} &= \frac{tg_0-qb_{n+1}t}{q(tx_{n+1}-x_n)} + \frac{1}{\underline{x_{n+1}}} - \frac{g_n-b_n}{x_n-x_{n+1}} - \frac{1}{\underline{x_n}} \\
	&= \frac{tg_0}{q(tx_{n+1}-x_1)} - \frac{b_{n+1}t}{tx_{n+1}-x_1} + \frac{b_na_{n+1}}{(x_n-x_{n+1})g_n} - \frac{a_{n+1}}{x_n-x_{n+1}}.
\end{split}\]
We can rewrite it to
\begin{equation}\begin{split}\label{eq:qFST_g_lem_proof_2}
	&b_i(x_{i-1}-x_i) + a_i(x_i-x_{i+1}) + (x_{i-1}-x_i)(x_i-x_{i+1})(y_{i-1}-y_i) \\
	&= (x_{i-1}-x_i)g_i + b_{i-1}a_i(x_i-x_{i+1})\frac{1}{g_{i-1}}\quad (i=1,\ldots,n), \\
	&b_{n+1}t(x_n-x_{n+1}) + a_{n+1}(tx_{n+1}-x_1) + (x_n-x_{n+1})(tx_{n+1}-x_1)(y_n-y_{n+1}) \\
	&= \frac tq(x_n-x_{n+1})g_0 + b_na_{n+1}(tx_{n+1}-x_1)\frac{1}{g_n}.
\end{split}\end{equation}
Combining equations \eqref{eq:qFST_g_lem_proof_1} and \eqref{eq:qFST_g_lem_proof_2}, we obtain
\begin{equation}\begin{split}\label{eq:qFST_g_lem}
	tf_1\frac{F_{n+1}}{F_1}\overline{g}_0 + tb_1a_2\frac{F_2}{F_1}\frac{1}{\overline{g}_1} &= tg_1 + qb_{n+1}a_1f_1\frac{1}{g_0}, \\
	f_i\frac{F_{i-1}}{F_i}\overline{g}_{i-1} + b_ia_{i+1}f_{i-1}\frac{F_{i+1}}{F_i}\frac{1}{\overline{g}_i} &= f_{i-1}g_i + b_{i-1}a_if_i\frac{1}{g_{i-1}}\quad (i=2,\ldots,n), \\
	t\frac{F_n}{F_{n+1}}\overline{g}_n + b_{n+1}a_1f_n\frac{F_1}{F_{n+1}}\frac{1}{\overline{g}_0} &= \frac tqf_ng_0 + b_na_{n+1}t\frac{1}{g_n}.
\end{split}\end{equation}

Since equation \eqref{eq:qFST_g_lem} is equivalent to \eqref{eq:qFST_g} in the case of $n=1$, we consider the case of $n\geq2$.
Then equation \eqref{eq:qFST_g_lem} is transformed to
\begin{equation}\begin{split}\label{eq:qFST_g_lem_linear_1}
	G_{n+1} - q^{n-1}t\left(t\frac{g_0g_1}{f_1}+qb_{n+1}a_1\right)G_1 + q^{n-1}b_1a_2t^2\frac{g_0g_1}{f_1}G_2 &= 0, \\
	G_{i-1} - \left(\frac{f_{i-1}g_{i-1}g_i}{f_i}+b_{i-1}a_i\right)G_i + b_ia_{i+1}\frac{f_{i-1}g_{i-1}g_i}{f_i}G_{i+1} &= 0\quad (i=2,\ldots,n), \\
	G_n - \left(\frac1qf_ng_ng_0+b_na_{n+1}\right)G_{n+1} + q^{n-1}b_{n+1}a_1tf_ng_ng_0G_1 &= 0,
\end{split}\end{equation}
via a transformation
\[\begin{split}
	\overline{g}_i &= \frac{F_{i+1}G_i}{g_iF_iG_{i+1}}\quad (i=1,\ldots,n), \\
	\overline{g}_0 &= \frac{1}{q^{n/2}t}\prod_{i=1}^{n}\frac{1}{\overline{g}_i} = \frac{1}{q^{n-1}t^2}\frac{F_1G_{n+1}}{g_0F_{n+1}G_1}.
\end{split}\]
Furthermore equation \eqref{eq:qFST_g_lem_linear_1} is reduced to a system of linear equations
\begin{equation}\label{eq:qFST_g_lem_mat}
	\begin{bmatrix}1&-\alpha_1&\beta_1\\&1&-\alpha_2&\beta_2\\&&1&-\alpha_3&\beta_3\\&&&&\ddots\\&&&&&-\alpha_{n-2}&\beta_{n-2}\\&&&&&1&-\alpha_{n-1}\\\beta_n&&&&&&1\end{bmatrix}\begin{bmatrix}G_1\\G_2\\G_3\\\vdots\\G_{n-2}\\G_{n-1}\\G_n\end{bmatrix} = \begin{bmatrix}0\\0\\0\\\vdots\\0\\-\beta_{n-1}G_{n+1}\\\alpha_nG_{n+1}\end{bmatrix},
\end{equation}
where
\[\begin{split}
	&\alpha_i = \frac{f_ig_ig_{i+1}}{f_{i+1}} + b_ia_{i+1},\quad
	\beta_i = b_{i+1}a_{i+2}\frac{f_ig_ig_{i+1}}{f_{i+1}}\quad (i=1,\ldots,n-1), \\
	&\alpha_n = \frac1qf_ng_ng_0 + b_na_{n+1},\quad
	\beta_n = q^{n-1}b_{n+1}a_1tf_ng_ng_0.
\end{split}\]
In fact equation \eqref{eq:qFST_g_lem_linear_1} can be rewritten to
\begin{equation}\begin{split}\label{eq:qFST_g_lem_linear_2}
	G_1 - b_1a_2G_2 &= \frac{qf_1}{tg_0g_1}\left(\frac{1}{q^nt}G_{n+1}-b_{n+1}a_1G_1\right), \\
	G_i - b_ia_{i+1}G_{i+1} &= \frac{f_i}{f_{i-1}g_{i-1}g_i}(G_{i-1}-b_{i-1}a_iG_i)\quad (i=2,\ldots,n), \\
	\frac{1}{q^nt}G_{n+1} - b_{n+1}a_1G_1 &= \frac{1}{q^{n-1}tf_ng_ng_0}(G_n-b_na_{n+1}G_{n+1}),
\end{split}\end{equation}
and the first equation of \eqref{eq:qFST_g_lem_linear_2} is derived from a combination of the other equations as follows.
\[\begin{split}
	\frac{1}{q^nt}G_{n+1} - b_{n+1}a_1G_1 &= \frac{1}{q^{n-1}tf_ng_ng_0}(G_n-b_na_{n+1}G_{n+1}) \\
	&= \frac{1}{q^{n-1}tf_{n-1}g_{n-1}g_n^2g_0}(G_{n-1}-b_{n-1}a_nG_n) \\
	&\,\,\,\vdots \\
	&= \frac{1}{q^{n-1}tf_1g_1g_2^2\ldots g_n^2g_0}(G_1-b_1a_2G_2) \\
	&= \frac{tg_0g_1}{qf_1}(G_1-b_1a_2G_2).
\end{split}\]
We will solve system \eqref{eq:qFST_g_lem_mat} for $G_i$ $(i=1,\ldots,n)$ in order to derive equation \eqref{eq:qFST_g}.
For this purpose we introduce the following lemma.

\begin{lem}\label{lem:eq_qFST_g_lem_det}
The determinant of the coefficient matrix of system \eqref{eq:qFST_g_lem_mat} is given by
\[
	\begin{vmatrix}1&-\alpha_1&\beta_1\\&1&-\alpha_2&\beta_2\\&&1&-\alpha_3&\beta_3\\&&&&\ddots\\&&&&&-\alpha_{n-2}&\beta_{n-2}\\&&&&&1&-\alpha_{n-1}\\\beta_n&&&&&&1\end{vmatrix} = 1 + q^{n/2}\sum_{j=1}^{n}b_{n+1}a_1\prod_{k=1}^{j-1}b_ka_{k+1}\frac{\prod_{l=j+1}^{n}g_l}{\prod_{l=1}^{j-1}g_l}f_j.
\]
\end{lem}

\begin{proof}
We obtain
\[
	\begin{vmatrix}1&-\alpha_1&\beta_1\\&1&-\alpha_2&\beta_2\\&&1&-\alpha_3&\beta_3\\&&&&\ddots\\&&&&&-\alpha_{n-2}&\beta_{n-2}\\&&&&&1&-\alpha_{n-1}\\\beta_n&&&&&&1\end{vmatrix} = 1 + (-1)^{n-1}\beta_n\Delta_{n-1}.
\]
where
\[
	\Delta_i = \begin{vmatrix}-\alpha_1&\beta_1\\1&-\alpha_2&\beta_2\\&1&-\alpha_3&\beta_3\\&&&\ddots\\&&&&-\alpha_{i-1}&\beta_{i-1}\\&&&&1&-\alpha_i\end{vmatrix},
\]
via a cofactor expansion.
The determinant of the tridiagonal matrix $\Delta_i$ satisfies a recurrence relation
\[
	\Delta_0 = 1,\quad
	\Delta_1 = -\alpha_1,\quad
	\Delta_i = -\alpha_i\Delta_{i-1} - \beta_{i-1}\Delta_{i-2}\quad (i=2,3,\ldots).
\]
By solving this relation, we obtain
\[
	\Delta_i = (-1)^i\sum_{j=1}^{i+1}\prod_{k=1}^{j-1}b_ka_{k+1}\prod_{l=j}^{i}g_lg_{l+1}\frac{f_j}{f_{i+1}}\quad (i=2,3,\ldots).
\]
It implies
\[\begin{split}
	(-1)^{n-1}\beta_n\Delta_{n-1} &= q^{n-1}b_{n+1}a_1tf_ng_ng_0\sum_{j=1}^{n}\prod_{k=1}^{j-1}b_ka_{k+1}\prod_{l=j}^{n-1}g_lg_{l+1}\frac{f_j}{f_n} \\
	&= q^{n/2}\sum_{j=1}^{n}b_{n+1}a_1\prod_{k=1}^{j-1}b_ka_{k+1}\frac{\prod_{l=j+1}^{n}g_l}{\prod_{l=1}^{j-1}g_l}f_j
\end{split}\]
\end{proof}

Thanks to Lemma \ref{lem:eq_qFST_g_lem_det} we can find that system \eqref{eq:qFST_g_lem_mat} admits only one solution for $G_i$ $(i=1,\ldots,n)$.
Hence we only have to verify that
\[
	G_i = \sum_{j=i}^{n}\prod_{k=i}^{j-1}b_ka_{k+1}\frac{\prod_{l=j+1}^{n}g_l}{\prod_{l=1}^{j-1}g_l}f_j + q^{n/2}t\prod_{k=i}^{n}b_ka_{k+1} + q^nt\sum_{j=1}^{i-1}\frac{b_{n+1}a_1\prod_{k=1}^{n}b_ka_{k+1}}{\prod_{k=j}^{i-1}b_ka_{k+1}}\frac{\prod_{l=j+1}^{n}g_l}{\prod_{l=1}^{j-1}g_l}f_j,
\]
satisfy system \eqref{eq:qFST_g_lem_mat}.
It is shown since
\[\begin{split}
	G_i - b_ia_{i+1}G_{i+1} &= \left(1-q^ntb_{n+1}a_1\prod_{k=1}^{n}b_ka_{k+1}\right)\frac{\prod_{l=i+1}^{n}g_l}{\prod_{l=1}^{i-1}g_l}f_i\quad (i=1,\ldots,n), \\
	\frac{1}{q^nt}G_{n+1} - b_{n+1}a_1G_1 &= \frac{1}{q^{n/2}}\left(1-q^ntb_{n+1}a_1\prod_{k=1}^{n}b_ka_{k+1}\right),
\end{split}\]
satisfy equation \eqref{eq:qFST_g_lem_linear_2}.
From the above we have derived equation \eqref{eq:qFST_g}.

In the last we prove equation \eqref{eq:qFST_h}.
The first equation of \eqref{eq:qFST_f_ab} is rewritten to
\[
	\overline{tx_{n+1}-x_1} = -\frac{(b_0-\overline{g}_0)(\overline{g}_0-a_1)}{\overline{g}_0}\frac{tx_{n+1}x_1}{tx_{n+1}-x_1},
\]
from which we obtain equation \eqref{eq:qFST_h} immediately.

\subsection{Relationship between two types of dependent variables}

In equation \eqref{eq:def_fg} we define the dependent variables $(f_i,g_i,h)$ as rational functions in $(\underline{x_j},y_j)$ with constraint \eqref{eq:qFST_constraint}.
Conversely, the dependent variables $(\underline{x_j},y_j)$ are given as rational functions in $(f_i,g_i,h)$.
We obtain
\[
	\frac{1}{\underline{x_1}} - \frac{1}{\underline{tx_{n+1}}} = -\frac{(b_0-g_0)(g_0-a_1)}{g_0h},
\]
from equation \eqref{eq:qFST_h} and
\[
	\frac{1}{\underline{x_{i+1}}} - \frac{1}{\underline{x_i}} = -\frac{t(b_i-g_i)(g_i-a_{i+1})}{f_ig_ih}.\quad (i=1,\ldots,n),
\]
from equation \eqref{eq:qFST_f_ab}.
They imply
\begin{equation}\begin{split}\label{eq:def_fg2x}
	\frac{1}{\underline{x_i}} &= \frac{t}{q-t}\frac{(b_0-g_0)(g_0-a_1)}{g_0h} \\
	&\quad + \frac{t}{q-t}\sum_{j=1}^{i-1}\frac{t(b_j-g_j)(g_j-a_{j+1})}{f_jg_jh} + \frac{q}{q-t}\sum_{j=i}^{n}\frac{t(b_j-g_j)(g_j-a_{j+1})}{f_jg_jh}\quad (i=1,\ldots,n+1).
\end{split}\end{equation}
We also obtain
\[
	y_i = -\frac{g_i/\underline{x_{i+1}}-a_{i+1}/\underline{x_i}}{g_i-a_{i+1}}\quad (i=1,\ldots,n),
\]
from equation \eqref{eq:qFST_f_a}.
The rest variable $y_{n+1}$ is given by constraint \eqref{eq:qFST_constraint}.

\begin{rem}
Equation \eqref{eq:def_fg2x} allows us to express the backward $q$-shifts $\underline{f_1},\ldots,\underline{f_n}$ as rational functions in $(f_i,g_i)$.
Also we can express the backward $q$-shifts $\underline{g_1},\ldots,\underline{g_n}$ as rational functions in $(f_i,g_i)$ by solving equation \eqref{eq:qFST_g_lem} for $g_1,\ldots,g_n$.
We don't give its detail here.
\end{rem}

We next state a relationship between the variables $(f_i,g_i,h)$ and the ones $(x_j,y_j)$.
Equation \eqref{eq:qFST_g_lem_proof_2} is rewritten to
\begin{equation}\begin{split}\label{eq:def_fg2y}
	\frac qty_{n+1} - y_1 &= t(g_1-b_1)\frac{1}{f_1h} + a_1\left(qb_{n+1}\frac{1}{g_0}-1\right)\frac1h \\
	y_{i-1} - y_i &= t(g_i-b_i)\frac{1}{f_ih} + a_it\left(b_{i-1}\frac{1}{g_{i-1}}-1\right)\frac{1}{f_{i-1}h}\quad (i=2,\ldots,n), \\
	y_n - y_{n+1} &= \frac tq(g_0-qb_{n+1})\frac1h + a_{n+1}t\left(b_n\frac{1}{g_n}-1\right)\frac{1}{f_nh}.
\end{split}\end{equation}
Recall that
\[
	tx_{n+1} - x_1 = h,\quad
	x_i - x_{i+1} = \frac{f_ih}{t}\quad (i=1,\ldots,n).
\]
Hence we can give the variables $(x_j,y_j)$ as rational functions in $(f_i,g_i,h)$.
Conversely, since equation \eqref{eq:def_fg2y} admits only one solution for $g_1,\ldots,g_n$, the variables $(f_i,g_i,h)$ are given as rational functions in $(x_j,y_j)$.
Then one constraint between the variables $(x_j,y_j)$ is obtained together.
We don't give their explicit formulas here.
Those facts allows us to express the backward $q$-shifts $(\underline{x_j},\underline{y_j})$ as functions in $(x_j,y_j)$.

\section{Affine Weyl group symmetry}\label{sec:Affine_Weyl}

As is seen in Section \ref{sec:Review_qFST}, the system $q$-$P_{(n+1,n+1)}$ is invariant under the action of the group of symmetries $\langle r_0,\ldots,r_{2n+1},\pi\rangle\simeq\widetilde{W}(A^{(1)}_{2n+1})$.
This action can be restricted to system \eqref{eq:qFST_f}, \eqref{eq:qFST_g}.

\begin{thm}
The birational transformations $r_0,\ldots,r_{2n+1},\pi$ defined by \eqref{eq:Aff_Wey_even}, \eqref{eq:Aff_Wey_odd} and \eqref{eq:Aff_Wey_rot} act on the dependent variables $f_i,g_i$ $(i=1,\ldots,n)$ as follows.
\begin{equation}\label{eq:aFST_sym_even}
	r_{2j-2}(f_i) = f_i,\quad
	r_{2j-2}(g_i) = g_i,
\end{equation}
for $j=1,\ldots,n+1$,
\begin{equation}\begin{split}\label{eq:aFST_sym_odd}
	&r_1(f_1) = f_1\frac{R_1^{a,a,a}}{R_1^{b,a,b}},\quad
	r_1(g_1) = g_1\frac{R_1^{b,a,a}}{R_1^{b,b,b}},\quad
	r_1(f_i) = f_i\frac{R_1^{b,a,a}}{R_1^{b,a,b}},\quad
	r_1(g_i) = g_i\quad (i\neq1), \\
	&r_{2j-1}(f_{j-1}) = f_{j-1}\frac{R_j^{b,a,b}}{R_j^{b,a,a}},\quad
	r_{2j-1}(g_{j-1}) = g_{j-1}\frac{R_j^{b,b,b}}{R_j^{b,a,a}},\quad
	r_{2j-1}(f_j) = f_j\frac{R_j^{a,a,a}}{R_j^{b,a,a}},\quad
	r_{2j-1}(g_j) = g_j\frac{R_j^{b,a,a}}{R_j^{b,b,b}}, \\
	&r_{2j-1}(f_i) = f_i,\quad
	r_{2j-1}(g_i) = g_i\quad (i\neq j-1,j), \\
	&r_{2n+1}(f_n) = f_n\frac{R_{n+1}^{b,a,b}}{R_{n+1}^{a,a,a}},\quad
	r_{2n+1}(g_n) = g_n\frac{R_{n+1}^{b,b,b}}{R_{n+1}^{b,a,a}},\quad
	r_{2n+1}(f_i) = f_i\frac{R_{n+1}^{b,a,a}}{R_{n+1}^{a,a,a}},\quad
	r_{2n+1}(g_i) = g_i\quad (i\neq n), 
\end{split}\end{equation}
for $j=2,\ldots,n$ and
\begin{equation}\begin{split}\label{eq:aFST_sym_rot}
	&\pi(f_i) = \frac{q^2}{t}\frac{(g_iR_i^{*}-b_{i+1}R_{i+1}^{*})(b_{i+1}-g_{i+1})(R_{i+1}^{*}+1-\frac{t}{q})f_1}{(g_0R_0^{*}-b_1R_1^{*})(b_1-g_1)(R_1^{*}+1-\frac{t}{q})f_{i+1}},\quad
	\pi(g_i) = \frac{a_{i+1}}{q^{\rho_1}}\frac{b_{i+1}R_{i+1}^{*}}{g_iR_i^{*}}\quad (i\neq n), \\
	&\pi(f_n) = q\frac{(g_nR_n^{*}-b_{n+1}R_0^{*})(b_0-g_0)(R_0^{*}+1-\frac{t}{q})f_1}{(g_0R_0^{*}-b_1R_1^{*})(b_1-g_1)(R_1^{*}+1-\frac{t}{q})f_0},\quad
	\pi(g_n) = \frac{a_{n+1}}{q^{\rho_1}}\frac{b_{n+1}R_0^{*}}{g_nR_n^{*}},
\end{split}\end{equation}
where
\[\begin{split}
	R_j^{\alpha,\beta,\gamma} &= (g_j-\alpha_j)\frac{1}{f_j} + \left(\beta_jb_{j-1}\frac{1}{g_{j-1}}-\gamma_j\right)\frac{1}{f_{j-1}}\quad (j\neq n+1), \\
	R_{n+1}^{\alpha,\beta,\gamma} &= \frac{1}{q}(g_0-q\alpha_{n+1}) + \left(\beta_{n+1}b_n\frac{1}{g_n}-\gamma_{n+1}\right)\frac{1}{f_n},
\end{split}\]
and
\[\begin{split}
	R_i^{*} &= -\frac{f_i}{b_i-g_i}\left(\frac{t}{q}\sum_{j=1}^{i}R_j^{b,a,a}+\sum_{j=i+1}^{n+1}R_j^{b,a,a}\right) \\
	&= \frac{f_i}{b_i-g_i}\left(\sum_{j=0}^{i-1}\frac{\frac{t}{q}(1-\frac{a_{j+1}}{g_j})(b_j-g_j)}{f_j}+\frac{(\frac{t}{q}-\frac{a_{i+1}}{g_i})(b_i-g_i)}{f_i}+\sum_{j=i+1}^{n}\frac{(1-\frac{a_{j+1}}{g_j})(b_j-g_j)}{f_j}\right).
\end{split}\]
Here we set $f_0=t$.
\end{thm}

Recall that
\[
	b_0 = qb_{n+1},\quad
	g_0 = \frac{1}{q^{(n-2)/2}tg_1\ldots g_n},\quad
	q^{\rho_1} = (q^na_1b_1\ldots a_{n+1}b_{n+1})^{1/(n+1)}.
\]
In this section we prove this theorem.

\subsection{Proof of action \eqref{eq:aFST_sym_even}}

The action of $r_{2j-2}$ on the dependent variables with the exception of $y_{j-1}$ is trivial for any $j=1,\ldots,n+1$.
Hence we only have to verify that $r_{2j-2}(g_{j-1})=g_{j-1}$ is satisfied.
And it is shown by using the first equation of $q$-$P_{(n+1,n+1)}$ as follows.
\[\begin{split}
	r_{2j-2}(g_{j-1}) &= b_{j-1}\frac{\underline{x_j}}{\underline{x_{j-1}}}\frac{1+\underline{x_{j-1}}y_{j-1}-\frac{b_{j-1}-a_j}{x_{j-1}-x_j}\underline{x_{j-1}}}{1+\underline{x_j}y_{j-1}-\frac{b_{j-1}-a_j}{x_{j-1}-x_j}\underline{x_j}} \\
	&= b_{j-1}\frac{\underline{x_j}}{\underline{x_{j-1}}}\frac{a_j(1+\underline{x_{j-1}}y_{j-1})}{b_{j-1}(1+\underline{x_j}y_{j-1})} \\
	&= g_{j-1}.
\end{split}\]

\subsection{Proof of action \eqref{eq:aFST_sym_odd}}

The action of $r_{2j-1}$ on the dependent variables with the exception of $x_j$ is trivial for any $j=1,\ldots,n+1$.
Hence we have to investigate the following actions:
\[\begin{split}
	&r_1(f_1),\quad
	r_1(g_1),\quad
	r_1(f_i)\quad (i\neq1), \\
	&r_{2j-1}(f_{j-1}),\quad
	r_{2j-1}(g_{j-1}),\quad
	r_{2j-1}(f_j),\quad
	r_{2j-1}(g_j), \\
	&r_{2n+1}(f_n),\quad
	r_{2n+1}(g_n),\quad
	r_{2n+1}(f_i)\quad (i\neq n).
\end{split}\]

We first consider the action $r_{2j-1}(f_j)$.
It is described in terms of the variables $x_i,y_i$ as
\begin{equation}\label{eq:qFST_sym_odd_proof_1}
	r_{2j-1}(f_j) = t\frac{x_j-\frac{a_j-b_j}{y_{j-1}-y_j}-x_{j+1}}{tx_{n+1}-x_1} = t\frac{(x_j-x_{j+1})(y_{j-1}-y_j)-(a_j-b_j)}{(tx_{n+1}-x_1)(y_{j-1}-y_j)}.
\end{equation}
Substituting equation \eqref{eq:def_fg2y} to \eqref{eq:qFST_sym_odd_proof_1}, we obtain
\[
	r_{2j-1}(f_j) = f_j\frac{(g_j-a_j)\frac{1}{f_j}+a_j(b_{j-1}\frac{1}{g_{j-1}}-1)\frac{1}{f_{j-1}}}{(g_j-b_j)\frac{1}{f_j}+a_j(b_{j-1}\frac{1}{g_{j-1}}-1)\frac{1}{f_{j-1}}}\quad (j=1,\ldots,n+1).
\]
The other actions on the variables $f_1,\ldots,f_n$ can be shown in a similar way.

We next consider the action $r_{2j-1}(g_j)$.
By using the system $q$-$P_{(n+1,n+1)}$ twice, we can describe it in terms of the variables $x_i,y_i$ as
\begin{equation}\begin{split}\label{eq:qFST_sym_odd_proof_2}
	r_{2j-1}(g_j) &= a_{j+1}\frac{\underline{x_{j+1}}\left(1+\underline{x_j}y_j-\frac{a_j-b_j}{\underline{y_{j-1}}-\underline{y_j}}y_j\right)}{\left(\underline{x_j}-\frac{a_j-b_j}{\underline{y_{j-1}}-\underline{y_j}}\right)(1+\underline{x_{j+1}}y_j)} \\
	&= g_j\frac{a_j\underline{x_j}(y_{j-1}-y_j)}{b_j(1+\underline{x_j}y_{j-1})-a_j(1+\underline{x_j}y_j)} \\
	&= g_j\frac{(x_{j-1}-x_j)(y_{j-1}-y_j)}{(x_{j-1}-x_j)(y_{j-1}-y_j)-(a_j-b_j)(b_{j-1}\frac{1}{g_{j-1}}-1)}.
\end{split}\end{equation}
Then, substituting equation \eqref{eq:def_fg2y} to \eqref{eq:qFST_sym_odd_proof_2}, we obtain
\[\begin{split}
	r_{2j-1}(g_j) = g_j\frac{(g_j-b_j)\frac{1}{f_j}+a_j(b_{j-1}\frac{1}{g_{j-1}}-1)\frac{1}{f_{j-1}}}{(g_j-b_j)\frac{1}{f_j}+b_j(b_{j-1}\frac{1}{g_{j-1}}-1)\frac{1}{f_{j-1}}}\quad (j=1,\ldots,n+1).
\end{split}\]
The other actions on the variables $g_1,\ldots,g_n$ can be shown in a similar way.

\subsection{Proof of action \eqref{eq:aFST_sym_rot}}

The action $\pi(f_i)$ is rewritten by the second equation of $q$-$P_{(n+1,n+1)}$ and equation \eqref{eq:qFST_f_b} to
\begin{equation}\begin{split}\label{eq:qFST_sym_rot_proof_1}
	\pi(f_i) &= \frac{q^2}{t}\frac{\underline{y_i}-\underline{y_{i+1}}}{\underline{y_0}-\underline{y_1}} = \frac{q^2}{t}\frac{\underline{x_1}}{\underline{x_{i+1}}}\frac{\frac{g_if_i}{b_i-g_i}y_i-\frac{b_{i+1}f_{i+1}}{b_{i+1}-g_{i+1}}y_{i+1}}{\frac{g_0t}{b_0-g_0}y_0-\frac{b_1f_1}{b_1-g_1}y_1}\quad (i=1,\ldots,n-1), \\
	\pi(f_n) &= \frac{q^2}{t}\frac{\underline{y_n}-\underline{y_{n+1}}}{\underline{y_0}-\underline{y_1}} = q\frac{\underline{x_1}}{\underline{x_0}}\frac{\frac{g_nf_n}{b_n-g_n}y_n-\frac{b_{n+1}t}{b_0-g_0}y_0}{\frac{g_0t}{b_0-g_0}y_0-\frac{b_1f_1}{b_1-g_1}y_1}.
\end{split}\end{equation}
On the other hand, we obtain
\begin{equation}\begin{split}\label{eq:qFST_sym_rot_proof_2}
	\frac{(\frac{q}{t}-1)(x_0-x_1)}{q}\frac{1}{\underline{x_i}} &= \sum_{j=0}^{i-1}\frac{\frac{t}{q}(1-\frac{a_{j+1}}{g_j})(b_j-g_j)}{f_j} + \sum_{j=i}^{n}\frac{(1-\frac{a_{j+1}}{g_j})(b_j-g_j)}{f_j}\quad (i=0,\ldots,n),
\end{split}\end{equation}
from equation \eqref{eq:qFST_f_ab} and
\begin{equation}\begin{split}\label{eq:qFST_sym_rot_proof_3}
	\frac{(\frac{t}{q}-1)(x_0-x_1)}{t}y_i &= \sum_{j=0}^{i-1}\frac{\frac{t}{q}(1-\frac{a_{j+1}}{g_j})(b_j-g_j)}{f_j} \\
	&\quad + \frac{(\frac{t}{q}-\frac{a_{i+1}}{g_i})(b_i-g_i)}{f_i} + \sum_{j=i+1}^{n}\frac{(1-\frac{a_{j+1}}{g_j})(b_j-g_j)}{f_j}\quad (i=0,\ldots,n),
\end{split}\end{equation}
from equation \eqref{eq:def_fg2y}.
Substituting equations \eqref{eq:qFST_sym_rot_proof_2} and \eqref{eq:qFST_sym_rot_proof_3} to \eqref{eq:qFST_sym_rot_proof_1}, we can show the action $\pi(f_i)$.

The action $\pi(g_i)$ is rewritten by equation \eqref{eq:qFST_f_b} to
\begin{equation}\begin{split}\label{eq:qFST_sym_rot_proof_4}
	\pi(g_i) &= \frac{b_{i+1}}{q^{\rho_1}}\frac{y_{i+1}(1+\underline{x_i}y_i)}{y_i(1+\underline{x_{i+1}}y_{i+1})} = \frac{a_{i+1}}{q^{\rho_1}}\frac{\frac{b_{i+1}f_{i+1}}{b_{i+1}-g_{i+1}}y_{i+1}}{\frac{g_if_i}{b_i-g_i}y_i}\quad (i=1,\ldots,n-1), \\
	\pi(g_n) &= \frac{b_{n+1}}{q^{\rho_1}}\frac{y_{n+1}(1+\underline{x_n}y_n)}{y_n(1+\underline{x_{n+1}}y_{n+1})} = \frac{a_{n+1}}{q^{\rho_1}}\frac{\frac{b_{n+1}t}{b_0-g_0}y_0}{\frac{g_nf_n}{b_n-g_n}y_n}.
\end{split}\end{equation}
Substituting equation \eqref{eq:qFST_sym_rot_proof_3} to \eqref{eq:qFST_sym_rot_proof_4}, we can show the action $\pi(g_i)$.

\section*{Acknowledgement}

The author would like to express his gratitude to Professors Masatoshi Noumi and Yasuhiko Yamada for helpful comments and advices.
This work was supported by JSPS KAKENHI Grant Number 15K04911.



\begin{thebibliography}{99}
\bibitem{FS1}
	K. Fuji and T. Suzuki,
	{\it Drinfeld-Sokolov hierarchies of type $A$ and fourth order Painlev\'{e} systems},
	Funkcial. Ekvac. \textbf{53} (2010) 143--167.
\bibitem{FS2}
	T. Suzuki and K. Fuji,
	{\it Higher order Painlev\'{e} systems of type $A$, Drinfeld-Sokolov hierarchies and Fuchsian systems},
	RIMS Kokyuroku Bessatsu \textbf{B30} (2012) 181--208.
\bibitem{G}
	R. Garnier,
	{\it Sur des \'{e}quations diff\'{e}rentielles du troisi\'{e}me ordre dont l'int\'{e}grale est uniform et sur une classe d'\'{e}quations nouvelles d'ordre sup\'{e}rieur dont l'int\'{e}grale g\'{e}n\'{e}rale a ses point critiques fix\'{e}s},
	Ann. Sci. \'{E}cole Norm. Sup. \textbf{29} (1912) 1--126.
\bibitem{JS}
	M. Jimbo and H. Sakai,
	{\it A $q$-analog of the sixth Painlev\'{e} equation},
	Let. Math. Phys. \textbf{38} (1996) 145--154.
\bibitem{K}
	H. Kawakami,
	{\it Matrix Painlev\'{e} systems},
	J. Math. Phys. \textbf{56} (2015) 033503.
\bibitem{KNY}
	K. Kajiwara, M. Noumi and Y. Yamada,
	{\it Discrete dynamical systems with $W(A^{(1)}_{m-1}\times A^{(1)}_{n-1})$ symmetry},
	Lett. Math. Phys. \textbf{60} (2002) 211--219.
\bibitem{M1}
	T. Masuda,
	{\it On the rational solutions of $q$-Painlev\'e V equation},
	Nagoya Math. J. \textbf{169} (2003) 119--143.
\bibitem{M2}
	T. Masuda,
	{\it A $q$-analogue of the higher order Painlev\'e type equations with the affine Weyl group symmetry of type $D$},
	Funkcial. Ekvac. \textbf{58} (2015) 405--430.
\bibitem{NY}
	H. Nagao and Y. Yamada,
	{\it Study of $q$-Garnier system by Pad\'e method},
	arXiv:1601.01099.
\bibitem{Sak1}
	H. Sakai,
	{\it A $q$-analog of the Garnier system},
	Funkcial. Ekvac. \textbf{48} (2005) 237--297.
\bibitem{Sak2}
	H. Sakai,
	{\it Isomonodromic deformation and 4-dimensional Painlev\'{e} type equations},
	UTMS \textbf{2010-17} (Univ. of Tokyo, 2010) 1--21.
\bibitem{Sas}
	Y. Sasano,
	{\it Higher order Painlev\'{e} equations of type $D^{(1)}_l$},
	RIMS Koukyuroku \textbf{1473} (2006) 143--163.
\bibitem{Su1}
	T. Suzuki,
	{\it A particular solution of a Painlev\'{e} system in terms of the hypergeometric function ${}_{n+1}F_n$},
	SIGMA \textbf{6} (2010) 078.
\bibitem{Su2}
	T. Suzuki,
	{\it A class of higher order Painlev\'{e} systems arising from integrable hierarchies of type $A$},
	AMS Contemp. Math. \textbf{593} (2013) 125--141.
\bibitem{Su3}
	T. Suzuki,
	{\it Six-dimensional Painlev\'{e} systems and their particular solutions in terms of rigid systems},
	J. Math. Phys. \textbf{55} (2014) 102902.
\bibitem{Su4}
	T. Suzuki,
	{\it A q-analogue of the Drinfeld-Sokolov hierarchy of type $A$ and $q$-Painlev\'{e} system},
	AMS Contemp. Math. \textbf{651} (2015) 25--38.
\bibitem{Tak}
	T. Takenawa,
	{\it Weyl group symmetry of type $D^{(1)}_5$ in the $q$-Painlev\'e V equation},
	Funkcial. Ekvac. \textbf{46} (2003) 173--186.
\bibitem{Tsu1}
	T. Tsuda,
	{\it On an integrable system of $q$-difference equations satisfied by the universal characters: its Lax formalism and an application to $q$-Painlev\'e equations},
	Comm. Math. Phys. \textbf{293} (2010) 347--359.
\bibitem{Tsu2}
	T. Tsuda,
	{\it UC hierarchy and monodromy preserving deformation},
	J. Reine Angew. Math. \textbf{690} (2014) 1--34.
\bibitem{Tsu3}
	T. Tsuda,
	{\it Hypergeometric solution of a certain polynomial Hamiltonian system of isomonodromy type},
	Quart. J. Math. \textbf{63} (2012) 489--505.
\end{thebibliography}
\end{document}